\documentclass[final,2p,times,numafflabel]{elsarticle}
\date{}
\usepackage[utf8]{inputenc}

\usepackage{amssymb}
\usepackage{amsmath,amsthm}
\usepackage{amsfonts}
\usepackage{graphicx}
\usepackage[colorlinks=false, allcolors=blue,hidelinks]{hyperref}
\newcommand\unaryminus{\smash{\scalebox{0.35}[1.0]{\( - \)}}}
\newtheorem{theorem}{Theorem}

\theoremstyle{definition}
\newtheorem{definition}{Definition}

\newtheorem{observation}{Observation}

\newcommand{\cav}{{\large{\textsc{cms}}}}

\usepackage{enumitem}

\usepackage[table,xcdraw]{xcolor}
\usepackage{todonotes}
\usepackage{physics}
\usepackage{amsmath}
\usepackage{tikz}
\usepackage{mathdots}
\usepackage{yhmath}
\usepackage{cancel}
\usepackage{color}
\usepackage{siunitx}
\usepackage{array}
\usepackage{multirow}
\usepackage{amssymb}
\usepackage{tabularx}
\usepackage{booktabs}
\usetikzlibrary{fadings}
\usetikzlibrary{patterns}
\usetikzlibrary{shadows.blur}
\usetikzlibrary{shapes}
\usepackage{adjustbox}
\usepackage{multicol}
\usepackage{cleveref}[2012/02/15]
\Crefname{observation}{Observation}{Observations}

\crefformat{footnote}{#2\footnotemark[#1]#3}

\let\ACMmaketitle=\maketitle
\renewcommand{\maketitle}{\begingroup\let\footnote=\thanks \ACMmaketitle\endgroup}

\usepackage{bold-extra}

\journal{the Special Issue on Economics and Computation, in Information Processing Letters.}

\begin{document}

\begin{frontmatter}



\title{On the Tractability Landscape of\\ the Conditional Minisum Approval Voting Rule} 


\author[1]{Georgios Amanatidis}
\author[2]{Michael Lampis}
\author[1]{\\Evangelos Markakis}
\author[3]{Georgios Papasotiropoulos}

\affiliation[1]{organization={Athens University of Economics and Business; Archimedes / Athena RC, Athens}, country={Greece}}
\affiliation[2]{organization={LAMSADE, CNRS, Université Paris Dauphine -- PSL, Paris}, country={France}}
\affiliation[3]{organization={University of Warsaw, Warsaw}, country={Poland}}

\begin{abstract}
{\small{This work examines the Conditional Approval Framework for elections involving multiple interdependent issues, specifically focusing on the Conditional Minisum Approval Voting Rule. We first conduct a detailed analysis of the computational complexity of this rule, demonstrating that no approach can significantly outperform the brute-force algorithm under common computational complexity assumptions and various natural input restrictions. In response, we propose two practical restrictions (the first in the literature) that make the problem computationally tractable and show that these restrictions are essentially tight. Overall, this work provides a clear picture of the tractability landscape of the problem, contributing to a comprehensive understanding of the complications introduced by conditional ballots and indicating that conditional approval voting can be applied in practice, albeit under specific conditions.
}}
\end{abstract}





\end{frontmatter}

\section{Introduction}

From everyday scenarios to political elections and from recommendation systems to information retrieval, the topics at hand are often interconnected. Alvin, actively participating in a participatory budgeting election, supports the implementation of at least one public project in his neighborhood, without specific liking. Betty, voting for MPs, believes that it is crucial to elect at least one representative from a particular minority group. Charlie prefers his favorite music app to suggest songs of a different style each day. Daphne wants her news search results to include perspectives from both radical and moderate viewpoints. Elmer prefers his social media feed to prioritize international content only after he has seen enough local updates. Fred, planning his next vacation with friends, is open to either a short trip or an overseas adventure, but prefers to avoid brief overseas trips. In such contexts, it is imperative to use voting frameworks that balance expressiveness and conciseness. Voters shouldn't have to detail all their preferences extensively, but the input must be expressive enough to reflect voters' true preferences towards yielding satisfactory outcomes. The Conditional Approval Framework, introduced by Barrot and Lang \cite{bl16}, aims at combining these qualities which makes it well-suited for multi-issue elections where voters' preferences on one issue depend on the outcome of others. 

The Conditional Approval Framework has lately received significant attention, leading to several works that build upon it. 
Specifically, Barrot and Lang \cite{bl16} proposed and axiomatically studied two voting methods within this framework, representing the utilitarian and egalitarian approaches, respectively. The former, the minisum method, or \textit{Conditional Minisum Approval Voting Rule} (\cav), extends classic approval voting to conditional multi-issue elections. This will be the focus of our work. Later research by Markakis and Papasotiropoulos \cite{mp20} focused on designing approximation algorithms to tackle the computational difficulty of the rule, which was already identified in \cite{bl16}. Notably, in a follow-up work, Markakis and Papasotiropoulos \cite{mp21} established a condition for polynomial-time computation of \cav. A complete analysis of the complexity of strategic control problems under the rule was also provided in \cite{mp21}. Additionally, Brill et al.~\cite{bmpp} examined proportionality questions around the framework, shifting focus from the minisum method to adaptations of proportional rules for the conditional setting. Specifically, they explored the conditional analogs of the \textit{Proportional Approval Voting} rule (PAV) \cite{Thie95a,Jans16a} and the \textit{Method of Equal Shares} (MES) \cite{peters2020proportionality} from the perspectives of axiomatic satisfiability and computational complexity. For an overview of conceptually similar frameworks and approaches, albeit not for the approval setting, we refer to the classic survey of Lang and Xia \cite{survey}. 

The conditional approval framework is promising and has the potential to improve upon the widely-used standard approval voting setting for multi-issue elections. Conditional approval voting offers greater expressiveness for voters' ballots while still accommodating those who would prefer the simplicity of classic approval voting. However, a major challenge lies in the computational complexity of determining the winner, hindering widespread adoption of the framework. Unlike many cases in Computational Social Choice where intractability stems from the complicated nature of a voting rule itself, here it arises from the vast number of potential outcomes.
Hence, even under the simplest and most natural rule, (the conditional analog of) approval voting rule, the winner determination problem is NP-hard \cite{bl16}. Driven by this, our focus returns to the approval voting rule (also referred to as the minisum solution) \cav, which is perceived as the foundational voting mechanism within the framework. We address the most fundamental computational problem related to the rule: determining the optimal outcome under \cav\ given voters' conditional preferences. In our work we provide a clean picture of the tractability landscape, contributing towards the comprehensive understanding of the intricacies that conditional ballots introduce.

\subsection{Our Contribution}
This paper aims at providing a clear indication that conditional approval voting can be applied in practice, albeit under specific conditions. The findings from the first half of our work highlight the computational limitations of the conditional setting, together with the necessity of concrete ballot restrictions to achieve efficient winner determination algorithms under the approval voting rule. Then, the second half demonstrates that certain restrictions can indeed be effectively established without completely compromising the framework's advantages. Overall, our study provides a pathway for applying conditional approval voting in actual elections, allowing for the consideration of issues that are interdependent according to the viewpoints of the voters.

More precisely, the first part of our work (\Cref{sec:hardness}) can be interpreted as a critique of the conditional framework, revealing that its greatest strength—voters' ample power of expressivity—also turns into its primary weakness: severe computational intractability. In particular, we provide a set of results that strengthen the existing NP-hardness by mainly demonstrating that finding the optimal outcome under \cav\ is not only computationally difficult but also that no algorithm exists that would solve the problem significantly faster than a brute-force exhaustive search of all possible outcomes. Importantly, the hardness holds even in very restricted families of instances. Our findings highlight that despite the natural appeal of the expressiveness of the conditional framework and the simplicity of the utilitarian (i.e., minisum) approach, unfortunately, \cav\ is deemed impractical for real-world usage, at least without substantial restrictions.

As a response to the aforementioned hardness, in the second part of the work (\Cref{sec:pos}) we propose concrete restrictions, which make the computation of the rule's outcome feasible in polynomial time. We argue that these restrictions are well-motivated and suitable for real elections and that are not overly stringent, allowing the conditional framework to continue offer greater expressive power than the---most commonly used---unconditional approval setting. We propose restrictions in two (orthogonal) directions: limiting the types of approval ballots (in \Cref{subsec:pos1}) and constraining the dependencies declared by the voters (in \Cref{subsec:pos2}). Notably, these restrictions can be easily enforced by election organizers. This stands in sharp contrast to the only known positive result for the rule in the literature: the polynomial algorithm suggested in \cite{mp21} pertains to instances where the union of voters' dependency graphs yields a graph of bounded treewidth, which, while theoretically intriguing, lacks practical applicability as such a restriction cannot be directly applied to the individual ballots per se. Our positive results hold for instances where conditional ballots are \textit{group-dichotomous}, or when the dependency graph of each voter has a \textit{bounded vertex cover number}. At a high level, group-dichotomous ballots express preferences where voters specify groups of issues they prefer to implement altogether or exclude from implementation altogether.
Vertex cover number is a structural graph parameter that is widely studied in parameterized algorithmics, but, crucially, in the literature of voting our work is among the first that leverages it towards providing tractability results. Finally, we highlight that
our positive findings are essentially tight.

This work examines the feasibility of conditional approval voting (CAV) under computational constraints, aiming to identify when and how it can be practically applied in elections. The first part highlights a fundamental challenge: while CAV allows voters more nuanced preferences by permitting interdependent issues, this flexibility leads to severe computational intractability. Specifically, the paper demonstrates that determining the optimal outcome under CAV is prohibitively hard, even with restricted instances. In response, the second part introduces practical limitations on voter ballots and dependency relationships that ensure polynomial-time computation of results without significantly sacrificing expressivity. By proposing constraints like group-dichotomous ballots and limiting dependency graph structures, this study offers solutions to make CAV more viable for real-world use.

\section{Preliminaries}
\label{sec:prelims}

Consider a collection of $m$ issues $I = \{I_1, \dots, I_m\}$, with each issue $I_j$ being associated with a \textit{domain}, which is a finite set $D_j$ of possible alternatives. Let $V = \{1, \dots, n\}$ denote a group of voters. We consider elections where voters in $V$ should make a decision on the issues of $I$. An \textit{outcome of an election} (or simply \textit{outcome}) is defined as a specific assignment of an alternative from $D_j$ to each issue $I_j$ from $I$. The specified value for an issue under an outcome will be called \textit{the outcome of the issue}. We say that $D = D_1 \times D_2 \times \dots \times D_m$ represents the set of all possible election outcomes and we use $d$ to denote the maximum domain size of the issues, i.e., $d:=\max_{i\in [m]}|D_i|$.

\subsection{Voting Format} \label{subsec:format}
In our work we follow the framework suggested in \cite{bl16} and examined in \cite{mp20,mp21,bmpp} adhering closely to their notation; we refer to those works for illustrating examples. Each voter $i \in [n]$ is associated with a directed graph $G_i = (I, E_i)$, known as the \textit{voter's dependency graph}, where the vertices correspond to the set of issues of the considered election and a directed edge $(I_k, I_j)$ belonging to the edge-set $E_i,$ implies that according to $i$ the outcome of issue $I_j$ is dependent on the outcome of issue $I_k$. Furthermore, let $N^{\unaryminus}_{i}(I_j)$ represent the (possibly empty) set of immediate in-neighbors of issue $I_j$ in $G_i$. An essential parameter for what follows is the maximum in-degree of the dependency graph $G_i$ of a voter $i$, denoted as $\Delta_i^{\unaryminus} := \max_{j \in [m]}\{|N^{\unaryminus}_{i}(I_j)|\}$. Say that $\Delta$ is the maximum in-degree among all voters' dependency graphs, so, $\Delta:=\max_{i \in [n]}\Delta_i^{\unaryminus}$. We could assume that $G_i$ is given either explicitly by voter $i$ when casting a ballot or is implicitly inferred from the ballot itself.

We will now outline the process by which voters submit their preferences before moving to the formal definition of the ballot format. For any issue $I_j$ with no in-neighbors in $G_i$ (i.e., its in-degree is zero), voter $i$ submits an unconditional approval ballot, equivalent to classic approval ballot, indicating all approved alternatives in $D_j$. If issue $I_j$ has one or more in-neighbors in $G_i$, let $N^{\unaryminus}_i(I_j)=\{I_{j_1}, I_{j_2}, \dots, I_{j_k} \}$. In this case, voter $i$ must list all the approved combinations in the form $\{\texttt{Pre}(I_j):A(I_j)\}$, where $A(I_j) \subseteq D_j$ and $\texttt{Pre}(I_j) \in D_{j_1} \times D_{j_2} \times \dots \times D_{j_k}$. Each combination $\{\texttt{Pre}(I_j):A(I_j)\}$ submitted by the voter $i$, signifies that the voter is satisfied with respect to issue $I_j$, provided the election outcome includes all alternatives in $\texttt{Pre}(I_j)$ (to be called a \textit{premise} for $I_j$) as well as an alternative from the \textit{approval set} $A(I_j)$ for $I_j$. Both scenarios, whether an issue has zero or positive in-degree, are unified in the following definition of \textit{conditional approval ballots}.

\begin{definition}
\label{def:conditionals}
	A conditional approval ballot of a voter $i$ over issues $I=\{I_1,\dots,I_m\}$ with domains $D_1,\dots, D_m$ respectively, is a pair $B_i=\langle G_i,\{C_{ij}, j\in [m]\} \rangle$, where $G_i$ is the dependency graph of voter $i$, and for each issue $I_j$, $C_{ij}$ is a set of conditional approval statements in the form $\{\texttt{Pre}(I_j):A(I_j)\},$ with $\texttt{Pre}(I_j)\in \prod_{k:I_k\in N^{\unaryminus}_{i}(I_j)} D_k$, and $A(I_j) \subseteq D_j$. 
\end{definition}

For any issue $I_j$ with an in-degree of zero as per some voter $i$, an approval vote for an alternative $\alpha_j \in D_j$ will simply be denoted as $\{\alpha_j\}$ (instead of $\{\emptyset: \alpha_j\}$). Furthermore, crucially, a voter is not obliged to approve an alternative in $D_j$ for every possible combination of values in $\prod_{k:I_k\in N^{\unaryminus}_{i}(I_j)} D_k$. As a result, we consider such a voter, say $i$, as a priori dissatisfied with respect to $I_j$ if the outcome of the issues in $N^{\unaryminus}_i(I_j)$ do not appear in the premise of the conditional ballot regarding issue $I_j$. 

Given a voter $i$ with a conditional ballot $B_i$, $B_i^j$ denotes the restriction of their ballot to issue $I_j$. Furthermore, a \textit{conditional approval voting profile} is represented by the tuple $P = (I, D, V, B)$, where $B = (B_1, B_2, \dots, B_n)$. We denote by $|P|$ the size of such a profile. Finally we define the \textit{global dependency graph} of a conditional approval voting profile as the undirected (simple) graph obtained by disregarding the direction (and multiplicities) of edges in the graph $(I, \bigcup_{i \in [n]} E_i)$, where $E_i$ represents the edge set of the dependency graph of voter $i$.

\subsection{Voting Rule} Our work is on the exploration of the generalization of the approval voting rule within the framework of conditional approval voting. We begin by introducing a measure for the dissatisfaction of a voter given an assignment of values to all issues. This measure is based on a generalized notion of the Hamming distance.

\begin{definition}
Given an outcome $r \in D$, we say that voter $i$ is dissatisfied (or disagrees) with the outcome of the issue $I_j$, say $r_j$, if for the projection of $r$ on $N^{\unaryminus}_i(I_j)$, say $t$, it holds that $\{t : r_j\}\notin B_i^j$, and is satisfied otherwise. We denote as $\delta_i(r)$ the total number of issues that dissatisfy voter $i$ under $r$.
\end{definition}

Given a conditional approval voting profile, \textit{Conditional Minisum Approval Voting rule} (\cav) outputs the outcome that minimizes the total number of disagreements over all voters. Equivalently for the objectives of our study, the outcome of an election under \cav\ maximizes total voters' satisfaction. 
If the global dependency graph of an instance is empty, i.e., $\Delta=0$, then the election degenerates to the approval rule in approval voting over multiple independent issues. We assume that the minimum degree in the global dependency graph is $1$, as issues that neither depend on nor affect other issues can be resolved by majority rule, regardless of the outcomes of the rest. To simplify notation, we will use \cav\ to refer both to the voting rule and to the related algorithmic problem; the exact meaning will always be clear from the context.
Formally, the algorithmic problem that we study is as follows. 

\begin{table}[h!]
	\centering
	\begin{tabular}{lp{0.65\columnwidth}}  
		\toprule
	 \multicolumn{2}{c}{\textsc{conditional minisum approval voting rule (\cav)} } \\
		\midrule
\textbf{Given:} & A voting profile $P$ with $m$ issues and $n$ voters casting conditional approval ballots and an integer $s$.\\
\textbf{Output:} & An outcome $r^* = (r_1^*,\dots, r_m^*)$ to all issues that achieves $\sum_{i \in [n]} \delta_i(r^*)\leq s$.\\
		\bottomrule
	\end{tabular}
\end{table}
\vspace{0.2cm}

Subsequently, we will make use of classic notions from the parameterized algorithms and complexity literature, such as the treewidth, pathwidth and vertex cover number of a graph. For the relevant definitions, we refer the reader to \cite{parameterized_book}.

\section{Intractability Results}
\label{sec:hardness}

We begin the section by noting that there is an obvious brute-force approach which optimally solves \cav\ by examining the dissatisfaction each possible outcome produces in the electorate and selecting the one that minimizes this value. Importantly, even though this result is trivial, it already underscores that \cav\ is a suitable choice for elections over a few interdependent issues. No better algorithm for the problem is known in the literature, with respect to worst-case running time complexity.

\begin{observation}
\label{obs:brute-force}
\cav\ is trivially solvable in time $O(d^m\cdot |P|^{O(1)})$.
\end{observation}

 In this section, we demonstrate that an algorithm with a worst-case complexity significantly better than the trivial one is unattainable, under common computational complexity assumptions. We begin with a strong negative result, showing that it is hard to find an algorithm that runs in time whose exponential dependence on $m$ is even slightly reduced, for general instances. The construction in this proof is interesting in its own right as it reveals a key insight: that the computational difficulty of \cav\ stems from the same source as its power, namely the expressiveness of the ballots. Specifically, it shows that the broad definition of the setting allows voters to submit ballots that could in principle encode the constraints of computationally hard problems. 
 
 The proof of the hardness relies on the Strong Exponential Time Hypothesis (SETH) \cite{ImpagliazzoP01} which, roughly, states that for all $\epsilon>0$ there exists a $k$ such
that $k$-\textsc{sat} cannot be solved in time $O^*((2-\epsilon)^\nu)$, where $\nu$ is
the number of variables of the $k$-\textsc{sat} instance and $O^*(\cdot)$ suppresses factors polynomial in the input size. The SETH along with its weakening ETH (which roughly speaking says that 3-\textsc{sat}
cannot be solved in time subexponential in the number of variables) are commonly used conjectures in parameterized algorithmics as they often allow for tight results (up to small factors) which can provide a (almost) complete understanding of the complexity of the studied problems \cite{parameterized_book}.

\begin{theorem} \label{thm:1}If there exists an $\epsilon>0$ such that \cav\ can be solved in
time $d^{(1-\epsilon)m}\cdot|P|^{O(1)}$, then the SETH is false. \end{theorem} 

\begin{proof}

Assume we have a \cav\ algorithm running in time $d^{(1-\epsilon)m}|P|^c$,
where $c$ is a constant. We will obtain a \textsc{sat} algorithm which solves
$k$-\textsc{sat} instances on $\nu$ variables, for all fixed $k$, in time $O(2^{(1-\frac{\epsilon}{2})\nu})$,
falsifying the SETH.

Suppose that we are given a $k$-\textsc{sat} formula $\phi$ with $\nu$ variables. We construct a \cav\ instance where the number of issues will be
$m=\lceil\frac{3kc}{\epsilon}\rceil$. In order to do so, we fix an arbitrary partition of the $\nu$
variables into $m$ sets, each of size at most $\lceil \frac{\nu}{m}\rceil$, call
these sets of variables $S_1,\ldots, S_m$. Intuitively, each set will correspond to an issue and the outcome we select
for issue $I_j$ will encode the assignment to the variables of $S_j$, for
$j\in[m]$. We set the domain of each issue $I_j$ to contain
$2^{|S_j|}$ distinct outcomes, each associated with an assignment to
the variables of $S_j$, therefore corresponding to all bit strings of length $O(\log|S_j|)$.

We proceed with specifying the ballots of the voters aiming to ensure that it is possible to satisfy every voter with respect to all issues if
and only if $\phi$ is satisfiable. For each clause $c_i$ of $\phi$ we will add a single voter, to be called $v_i$. We note that $c_i$ contains variables from $k'\le k$ distinct
sets, say $S_{i_1}, S_{i_2},\ldots, S_{i_{k'}}$. Fix an arbitrary $j\in\{i_1,\ldots,i_{k'}\}$ and consider any tuple of outcomes for the issues in $\{I_t: t \in \{i_1,\ldots,i_{k'}\}\setminus\{j\}\}$. For each such tuple we infer an
assignment that covers all variables of $c_i$ except those of $S_{j}$. If this
assignment already satisfies the clause $c_i$, then we add to the ballot of $v_i$ a
conditional approval statement, with prefix this tuple and with approval set
the set of all possible outcomes for issue $I_j$. Otherwise, if this assignment does not
already satisfy $c_i$, then we add to the ballot of $v_i$ a conditional approval
statement, with prefix this tuple and with approval set the set of all outcomes for
issue $I_j$ which represent an assignment that satisfies $c_i$ using variables of
$S_j$. Finally, $v_i$ is unconditionally satisfied with any alternative of any other issue apart from $I_j$.
We also set $s=0$.

For the forward direction, we argue that if $\phi$ is satisfiable, then there is an outcome that
satisfies every voter with respect to all issues. Indeed, for each issue $I_j$, we can select the outcome that
corresponds to the value given to $S_j$ by the satisfying assignment of $\phi$. Consider
now $v_i$, let $\{i_1,\ldots,i_{k'}\}$ be
the indices of sets of variables that appear in $c_i$, as above, and consider an issue
$I_j\in\{i_1,\ldots,i_{k'}\}$. By the construction, the voter is satisfied for all the rest of the issues. Since the ballot of that voter includes a conditional approval statement for every tuple
of outcomes for issues $\{i_1,\ldots,i_{k'}\}\setminus\{j\}$, there exists a
conditional approval statement whose premise is satisfied in the selected outcome. If $c_i$ is
satisfied by some variable outside of $S_j$, then the conditional approval
statement allows any outcome for issue $I_j$, so the voter is already satisfied. If not,
then the assignment must satisfy $c_i$ via the values of variables in $S_j$, so
the voter is again satisfied by the selected outcome, with respect to issue $I_j$.

For the reverse direction, we begin by extracting an assignment to the formula $\phi$ from an
outcome that satisfies every voter with respect to all issues. To see that this assignment satisfies $\phi$, suppose the contrary, that clause $c_i$ has all its literals set to False. Then, $v_i$ would not be satisfied with respect to the only issue $I_j$ of interest; a contradiction.

Finally, we will show that the claimed running time is sufficient to falsify the SETH. The
number of outcomes we consider for each issue is at most $d=2^{|S_j|} \le
2^{\frac{\nu}{m}+1} \le 2\cdot 2^{\frac{\epsilon \nu}{3ck}}$. Furthermore, to represent the ballot of each voter we needed a size of $d^k=O(2^{\frac{\epsilon
\nu}{3c}})$, because each voter is interested in only at most $k$ issues, each of domain $d$. We also
note that the input formula can be assumed to have at most $O(\nu^k)$ clauses.
Hence $|P| =  O((\nu d)^k)$ and the $|P|^c$ factor in the running time of the claimed algorithm for \cav\ can be
upper-bounded by $2^{\frac{\epsilon \nu}{2}}$, assuming $\nu$ is
sufficiently large. On the other hand, the $d^{(1-\epsilon)m}$ factor in the
running time is at most $2^{(\frac{\nu}{m}+1)(1-\epsilon)m} =
O(2^{(1-\epsilon)\nu})$, because $m$ is a constant as it depends only on
$k,c,\epsilon$. Hence, the given $k$-\textsc{sat} instance can now be solved, in total running time of at most
$2^{(1-\frac{\epsilon}{2})\nu}$.
\end{proof}

The proof of \Cref{thm:1} has made evident the fact that \Cref{def:conditionals} does not impose restrictions on the size of the ballots cast by the voters. Consequently, these ballots could potentially be exponentially large due to the premise of an issue that can in principle take \( d^{m-1} \) different values. Therefore, to achieve polynomial computation for \cav, it is crucial to require voters to submit their ballots concisely.
Nevertheless, our subsequent hardness results demonstrate that various interpretations of “conciseness" alone cannot lead to polynomial-time solvability.

\Cref{thm:1} presents a very strong negative result in terms of the running time and, moreover, its proof highlights that the power of expressiveness the framework offers to the voters, as one would expect, comes at the cost of computational intractability. Despite this, the hope for the polynomial-time solvability of \cav\ in realistic instances is not yet lost. While significant, the result of \Cref{thm:1} seems primarily of theoretical interest, as the voters' ballots in the constructed instance are quite complex and it is unlikely to be encountered in practice. Additionally, the hardness result is based on the Strong Exponential Time Hypothesis (SETH), and a result based on the more widely accepted Exponential Time Hypothesis (ETH) would be more compelling. The following theorem, although weaker in terms of the running time bound---fortunately just marginal---not only holds under ETH but also demonstrates hardness using significantly simpler voter ballots.

\begin{theorem}
\label{thm:o(m)} There is no
$f(m)d^{o(m)}\cdot|P|^{O(1)}$-time algorithm for \cav, for any computable function $f$, assuming the Exponential Time Hypothesis, even if $\Delta=2$, and in particular if the dependency graph of each voter who has conditional preferences is an (in-)star on three vertices.
\end{theorem}

\begin{proof}
In $k$-\textsc{multicolored clique}, we are given an undirected graph $\mathcal{G}=(\mathcal{V},\mathcal{E})$ and a set of $k$ color classes for the vertices of $\mathcal{G}$. It is assumed that there are exactly $c$ vertices of each color, or otherwise one can extend a color class by adding isolated vertices of that color. The question is whether there exists a subgraph of $\mathcal{G}$ on $k$ vertices, each of a different color, that forms a clique. It is known that $k$-\textsc{multicolored clique} cannot be decided in time $f(k)c^{o(k)}$ for any computable function $f$, under ETH \cite{chen2006strong,parameterized_book}.

Given an instance $\Pi$ of $k$-\textsc{multicolored clique}, we create an instance $\Pi'$ of \cav. 
The construction will be such that the number of colored classes $k$ corresponds to the total number of issues $m - 1$, and the number of vertices of each color class $c$ corresponds to the number of alternatives for each issue $d$. Thus, a bound of $f(m) d^{o(m)}$ will follow. We describe $\Pi'$ below:

\begin{itemize}[leftmargin=10pt,parsep=4pt,topsep=10pt,itemsep=4pt]
\item There is a binary special issue $I_s$ with alternatives $\texttt{P}$ and $\texttt{N}$, as well as a special voter $v_s$ who approves $\texttt{P}$ and all alternatives of issues other than $I_s$, unconditionally.
\item There is one issue $I_j$ with a domain size $c$, for each color $j$ in $\Pi$. The $c$ alternatives for each issue $I_j$ correspond to the $c$ vertices that are colored under $j$ in $\Pi$.
\item For each pair of colors, we add a voter. For a pair of colors $x,y$, the voter that corresponds to these colors has the following preferences: for every edge $e = (u, v)\in \mathcal{E},$ joining a vertex $u$ of color $x$ with a vertex $v$ of color $y$, this voter approves all possible alternatives for issues other than $I_s$, and for issue $I_s$, this voter votes conditionally on issues $I_{x}$ and $I_{y}$ as follows: $\{uv: \texttt{P}\}$. 
\end{itemize}

We also set $s$, the decision variable of \cav, to $0$. In the created instance $\Pi'$, we have $\binom{k}{2} + 1$ voters and $k + 1$ issues, each with a domain size of $c$. Furthermore, the dependency graph of $v_s$ has $0$ edges, while the dependency graphs of the remaining voters each have $2$ edges, both directed towards $I_s$, resulting in a global dependency graph that is a star directed towards its center. It’s worth noting that each voter only requires a single conditional statement, which further underscores the naturalness of the reduction presented and the simplicity of the instance used.

Suppose $\Pi$ is a \texttt{yes}-instance. Then, there exists a clique on $k$ vertices such that each vertex is of a different color. We fix the following outcome for the corresponding instance $\Pi'$ of \cav:
for each issue $I_j$ corresponding to color $j$, we select the alternative that corresponds to the vertex of the clique colored under $j$. For the special issue $I_s$, we select $\texttt{P}$. 
We need to show that all voters are satisfied with respect to all issues. Trivially, the special voter $v_s$ is satisfied with the choice of $\texttt{P}$. Consider an arbitrary voter among the remaining $\binom{k}{2}$ voters, corresponding to the pair of colors $(x,y)$. By the existence of a multicolored clique, there is an edge in $\mathcal{E}$, as well as in the clique, joining a vertex of color $x$, say $u$, with a vertex of color $y$, say $v$. Consequently, this voter is satisfied with respect to $I_s$ because both $u$ and $v$, and $\texttt{P}$, are selected in the outcome. By the fact that all voters approve any alternative for issues other than $I_s,$ we have that this voter is satisfied with all $k + 1$ issues. Therefore, the selected outcome results in a total dissatisfaction of $0$ for the electorate, concluding the forward direction.

For the reverse direction, suppose $\Pi'$ is a \texttt{yes}-instance of \cav, meaning there exists an assignment satisfying all voters completely. We will show that the vertices of $\mathcal{G}$ corresponding to the alternatives of the satisfying assignment for issues other than $I_s$ form a multicolored clique in $\mathcal{G}$. The vertices corresponding to the selected alternatives for issues other than $I_s$ are exactly $k$ in number and of different colors, due to the reduction's construction. The condition $s = 0$ ensures that the chosen outcome includes $\texttt{P}$ for $I_s$, as otherwise $v_s$ would be dissatisfied. Furthermore, for a total dissatisfaction of $0$, all $\binom{k}{2}$ voters must be satisfied with respect to $I_s$ as well. Focusing on an arbitrary voter corresponding to the pair of colors $(x,y)$, satisfaction with all issues implies the existence of a pair of alternatives, say $(u, v)$, from different issues, such that the voter votes for $\{uv: \texttt{P}\}$. By construction, this vote indicates the existence of an edge $(u, v)$ in $\mathcal{E}$. Thus, for every voter, or equivalently for every pair of colors, there is an edge connecting vertices of the two colors, using only the vertices corresponding to alternatives selected in the outcome. Therefore, $\Pi$ is a \texttt{yes}-instance as well, and, consequently, an algorithm running in $d^{o(m)}\cdot |P|^{O(1)}$ for \cav\ could be used to decide $k$-\textsc{multicolored clique} in
$c^{o(k)}f(k)$ as well, which does not hold, assuming ETH.
\end{proof}

Given that both hardness results from \Cref{thm:1,thm:o(m)} apply to cases where \(\Delta>1\), we now shift our attention to the simplest case of conditional ballots (which at the same time is the most examined case in the related literature), where \(\Delta=1\), and demonstrate a negative result for those instances as well. 
On the downside, the result that follows is slightly weaker in terms of the running time bound.

\begin{theorem}
\label{thm:sqrt}
There is no $d^{o(m/\log m)}\cdot|P|^{O(1)}$-time algorithm for \cav, assuming the Exponential Time Hypothesis, even if $\Delta=1$, and in particular if the dependency graph of each voter who has conditional preferences is an (out-)star
on three vertices. \end{theorem}

\begin{proof}

We reduce from 2-\textsc{csp}, in which we are given a \textsc{csp} instance in the form of a set of variables, each taking values
over a domain $\Sigma$ and a set of $k$ constraints, with each constraint
involving $2$ variables. 
A well-known
result from \cite{Marx10} (which was recently given a more streamlined proof
\cite{SMPS24}) states that this problem cannot be solved in time
$f(k)\cdot |\Sigma|^{o(k/\log k)}$ for any computable function $f$, unless the ETH is false. Our high-level plan here is to construct a \cav\ instance $\Pi'$ from a 2-\textsc{csp}
instance with $k$ constraints $\Pi$ so that the number of issues in $\Pi$
is $m=O(k)$ and the maximum domain is $d=|\Sigma|^2$, while the reduction runs in
polynomial time. If we achieve these properties, then a \cav\ algorithm running
in time $d^{o(m/\log m)}\cdot|P|^{O(1)}$ would imply an algorithm for 2-\textsc{csp} that contradicts the result from \cite{Marx10}, falsifying the ETH.

Our construction is as follows: for each variable $v$ of $\Pi$ we construct an issue $I_v$ with $D_v=\Sigma$
and for each constraint involving variables $uv$ we construct an issue $I_{uv}$
with $D_{uv}$ being a list of all the $|\Sigma|^2$ possible outcomes of the corresponding constraint.
Observe that we have $O(k)$ issues, as we can assume that each variable of $\Pi$ belongs to at least one constraint and hence $\Pi$ has at most $2k$ variables. Furthermore $d=|\Sigma|^2$. For each constraint of $\Pi$ involving variables $u,v$ and each assignment $a,b$ to
these variables satisfying this constraint, we add in $\Pi'$ a voter whose role
is intuitively to check that if we selected $(a,b)$ as the outcome of $I_{uv}$
then we must have selected $a$ and $b$ as the outcomes of $I_u$ and $I_v$
respectively. More precisely, this voter unconditionally accepts any outcome for
all issues other than $I_u, I_v$. Regarding issues $I_u$ and $I_v$, the voter votes conditionally on the outcome of $I_{uv}$ for each of them, as follows: if the outcome of $I_{uv}$ is not
$(a,b)$ the voter accepts any outcome of $I_u, I_v$, and, otherwise, if the outcome of $I_{uv}$ is $(a,b)$, then only accepts outcome $a$
for $I_u$ and $b$ for $I_v$. We set $s=0$ and this completes the construction. Observe that the dependency graph of each voter is an (out-)star on three vertices, resulting in $\Delta=1$.

Say that $\Pi$ is a \texttt{yes}-instance. Then, we can select the corresponding outcomes for issues representing
variables and constraints and it is not hard to see that all voters are
satisfied with respect to all of the issues. Hence, $\Pi'$ is also a \texttt{yes}-instance. For the reverse direction, suppose that $\Pi'$ is a \texttt{yes}-instance, i.e., there is an outcome $r^*$ satisfying all voters with respect to all issues. From that, we can extract an assignment for $\Pi$ by considering
the outcomes of issues that correspond to the variables of $\Pi$. We claim that this assignment satisfy all the $k$ constraints of $\Pi$. If not, there should be a constraint on a pair of variables, say $u$ and $v$ that is not
satisfied by the selected assignment. But that would mean that the outcome of issue $I_{uv}$ (which
corresponds by construction to a satisfying assignment of the constraint) is
not consistent with the outcomes of $I_u, I_v$, so, then the voter we
constructed for this assignment would be dissatisfied with respect to at least one issue. This contradicts the fact that $r^*$ satisfies all voters with respect to all issues. \end{proof}

From \cite{mp21}, it is known that an algorithm with a running time of \( d^{o(tw)/\log(tw)}\cdot|P|^{O(1)} \) is impossible under the Exponential Time Hypothesis (ETH), where $tw$ is the treewidth of the global dependency graph of the given instance.
\Cref{thm:1,thm:o(m),thm:sqrt}, provide stronger bounds than this, in terms of running time. Moreover, we observe that several structural graph parameters associated with the global dependency graph (beyond treewidth) also do not yield improved running times. For instance, we consider pathwidth $(pw)$ and vertex cover number (\(vc\)) as those parameters have also significant value regarding the polynomial-time computability of \cav\ (see \Cref{subsec:pos2}) but, also, the result holds for others. The hardness with respect to these arises in a straightforward way from the fact that the number of vertices in the global dependency graph equals \(m\), and it is an upper bound for those parameters.
Thus, in terms of graph parameters, the hardness proofs from \Cref{thm:1,thm:o(m),thm:sqrt} are as strong as they could be. Furthermore, while our constructions proving hardness explicitly used the fact that $s=0$, this does not diminish the generality of our results. One could simply introduce a voter $x$ along with a set of $s$ issues $Y$, where $x$ is dissatisfied by all alternatives of issues in $Y$ and satisfied by all of the rest, while all other voters approve all alternatives for issues in $Y$; targeting for a dissatisfaction of at most $s$ using the reductions presented.

\section{Tractability Results}
\label{sec:pos}

The preceding section not only reaffirms the NP-hardness of \cav\ \cite{bl16} but it also highlights impossibilities around the existence of solutions beyond trivial methods for general instances. Despite these results, we continue being advocates of the framework due to its substantial advantages in using conditional preferences for real-world multi-issue elections. Therefore, we now turn our attention to specific cases where efficient solutions are not already deemed impossible by the negative results of \Cref{sec:hardness} and we explore two different approaches: limiting the types of ballots (\Cref{subsec:pos1}) and limiting the structure of declared dependencies (\Cref{subsec:pos2}). Each approach led to the identification of a natural and, crucially, tight restriction.

\subsection{Restricting the Ballot Types}
\label{subsec:pos1}

In this section, we focus on binary domains and, for simplicity, we will use $D_i=\{0_i,1_i\}$ for each $i\in [m]$. Even in this case, the problem remains NP-hard for $\Delta=1$ even when each voter disapproves of only one of the four possible outcomes for the two involved issues \cite{bl16}. Thus, a stronger, or at least differently flavored, restriction is necessary. We propose the restriction of group-dichotomous instances.

\begin{definition}
\label{def:monotone}
    A conditional ballot $\{\texttt{Pre}(I_j):A(I_j)\}$ of a voter $i$ regarding a binary issue $I_j$ is called \textit{group-dichotomous} if (a) in case $A(I_j)=0_j$ then it holds that $\texttt{Pre}(I_j)=(0_k)_{k\in N_i^{\unaryminus}(I_j)}$,  (b) in case $A(I_j)=1_j$ then it holds that $\texttt{Pre}(I_j)=(1_k)_{k\in N_i^{\unaryminus}(I_j)}$, and (c) in case $A(I_j)=\{0_j,1_j\}$ then it either holds that $\texttt{Pre}(I_j)=(0_k)_{k\in N_i^{\unaryminus}(I_j)}$ or that $\texttt{Pre}(I_j)=(1_k)_{k\in N_i^{\unaryminus}(I_j)}$. An instance of \cav\ is said to be \textit{group-dichotomous} if all the conditional ballots are group-dichotomous (and no restrictions are imposed on unconditional ballots).
\end{definition}

A group-dichotomous instance can arise when voters identify sets of complementary issues, i.e., issues that are desirable if implemented concurrently, but lack individual value if implemented alone. Alternatively, voters may also identify sets of issues they wish to veto together, indicating they would be satisfied if none of those issues were implemented. 
This assumption is clearer in the simplest case where $\Delta=1$. Then, group-dichotomous instances include either unconditional ballots or ballots of the form $\{0_i:\{0_j,1_j\}\}$ or $\{1_i:\{0_j,1_j\}\}$ or $\{0_i:0_j\}$ or $\{1_i:1_j\}$ (but not $\{0_i:1_j\}$ or $\{1_i:0_j\}),$ for pairs $i,j\in [m]$. 
To illustrate this in real-world scenarios, consider a voter in a participatory budgeting context who approves the construction of a park only if pedestrian infrastructure is also built around it; neither the park alone nor the pedestrian infrastructure alone satisfies her, nor does the absence of both. Alternatively, a voter voting for the formation of a committee, may wish to state that if candidate A is not included, then they also do not want candidate B included, which could be because they recognize that A and B work well together and could be more beneficial to be placed together in another committee.
Importantly, we complement our polynomial-time algorithm for those instances with an impossibility result stating that the restriction of group-dichotomous ballots cannot be dropped entirely while allowing for polynomial algorithms.

\begin{theorem}
\label{thm:monotone}
Given an instance of \cav\ on binary issues, if the instance is group-dichotomous then \cav\ can be optimally solved, otherwise it is NP-hard, even for $\Delta=1$.
\end{theorem}

\begin{proof}
For the tractability result we will first reduce a group-dichotomous instance $\Pi$ of \cav\ to an instance $\Pi'$ of \textsc{minimum constraint satisfaction (min-csp)}. In the \textsc{min-csp} problem we are given a collection of constraints on
(overlapping) sets of binary variables and the goal is to find an assignment of values to the variables that minimize the number of unsatisfied constraints. Say that $\Pi$ is an instance on $n$ voters and a set $I$ of $m$ (binary) issues, then we will create an instance $\Pi$ on $m$ (binary) variables with at most $nm$ constraints. Say that $x_j$ is the variable of $\Pi'$ that corresponds to issue $I_j$ of $\Pi$.
Fix a voter $i\in [n]$ and an issue $I_j\in I$. Say that $|N_i^{\unaryminus}(I_j)|=0$. We can assume that $i$ approves only one out of the two possible alternatives for $I_j$; or otherwise the considered pair will have the exact same contribution on any feasible solution. 
If $i$ votes for $0_j$ we add $\overline{x_j}$ as a constraint in $\Pi$ and if $i$ votes for $1_j$ we add $x_j$ as a constraint in $\Pi$. 
Now say that $|N_i^{\unaryminus}(I_j)|\geq 1$ and that the ballot of $i$ regarding issue
$I_j$ contains expressions of the form $\{\texttt{Pre}(I_j):A(I_j)\}$. 
Firstly say that $i$ approves only one option among $0_j$ and $1_j$ (conditioned on the outcomes of the issues in $N^{\unaryminus}_i(I_j)$) and not both.
If $A(I_j)=0_j$ then, by the fact that all conditional ballots are group-dichotomous, it holds that $\texttt{Pre}(I_j)=(0_k)_{k\in N^{\unaryminus}_i(I_j)}.$ We add the following constraint in $\Pi'$: $\overline{x_j}\wedge \bigwedge_{k\in N^{\unaryminus}_i(I_j)}\overline{x_k}$. The case of $A(I_j)=1_j$ is similar. If $i$ approves both $0_j$ and $1_j$ then, for each $j\in [m]$, we set $y_j:=\overline{x_j},$ if $\texttt{Pre}(I_j)=(0_k)_{k\in N^{\unaryminus}_i(I_j)}$ and $y_j:=x_j,$ if $\texttt{Pre}(I_j)=(1_k)_{k\in N^{\unaryminus}_i(I_j)},$ and we add the constraint: $(\overline{x_j}\wedge \bigwedge_{k\in N^{\unaryminus}_i(I_j)}y_k) \vee (x_j\wedge \bigwedge_{k\in N^{\unaryminus}_i(I_j)}y_k)\equiv \bigwedge_{k\in N^{\unaryminus}_i(I_j)}y_k$. Any assignment on the $m$ variables of $\Pi'$ that satisfies all but $s$ of the described constraints induces an assignment $r^*$ on the issues of $\Pi$ that achieves $\sum_{i\in[n]}\delta_i(r^*)= s$ as it dissatisfies exactly $s$ pairs of voters/issues and vice-versa. 
Each of the described constraints is expressed as a DNF-formula with at most two terms: one containing only positive literals and the other containing only negative literals. Such an instance of \textsc{min-csp} can be easily reduced to \textsc{min-cut}, hence being polynomially solvable. The presented reduction to \textsc{min-csp} is a generalization of the one in \cite{mp21}, for $\Delta>1$, with the added consideration of the format of the created constraints and utilizing the positive result from \cite{khanna1996optimization}.

For the intractability result, suppose that we can solve \cav\ in polynomial time for an instance that is not group-dichotomous. By following the reduction proving hardness from \cite{mp21}, we could solve the corresponding \textsc{min-csp} instance. But this instance would also have constraints that are not expressible as the aforementioned DNF-formulas. According to \cite{khanna2001approximability}, unless the constraints are as described, 
\textsc{min-csp} cannot be solved in polynomial time. 
\end{proof}

\subsection{Restricting the Structure of Dependencies}
\label{subsec:pos2}
This section focuses on restrictions on the dependencies declared by the voters that can result in polynomial-time algorithms. As detailed in \Cref{subsec:format}, these dependencies are illustrated by the dependency graph cast by each voter and, in turn, by the global dependency graph of the instance. The most definitive parameter for these graphs is the number of their vertices, or $m$, the number of issues in the given instance. For instances with relatively few issues, or, in other words, dependency graphs of few vertices, \cav\ is trivially polynomially solvable (cf. \Cref{obs:brute-force}). Here, we examine what lighter restrictions could be posed to solve \cav\ in polynomial time, for elections over a large number of issues. We focus on the natural case of $\Delta=1$ (again, as in \Cref{subsec:pos1}, the result admits a certain generalization but the necessary condition will be less natural in that case). Under the assumption of $\Delta=1$, it is known (see \cite{mp21}) that treewidth is a pivotal parameter that characterizes polynomial computation of \cav:

\begin{theorem}[\cite{mp21}]
    \label{characterization}
    Let $\mathbf{G}$ be a recursively enumerable (e.g., decidable) class of graphs, and let \cav$(\mathbf{G})$ be the class of instances of \cav\
with a global dependency graph that belongs to $\mathbf{G}$, and with
$\Delta = 1$. Assuming FPT$\neq$W[1], there is a polynomial-time algorithm for \cav$(\mathbf{G})$ if and only if every graph
in $\mathbf{G}$ has constant treewidth.
\end{theorem}

However, notice that the restriction on the treewidth that appears in \Cref{characterization} (originated from~\cite{mp21}) applies to the global dependency graph, meaning that in order to achieve a polynomial-time algorithm, we must ensure that the union of voters' dependency graphs has a certain structural guarantee. Enforcing this in practice is challenging---unless we a priori fix a common dependency graph for all voters (with bounded treewidth), nevertheless, this would severely limit the power of expressiveness of voters' ballots,
thereby weakening the framework.

As a solution, our focus here is on proposing a \textit{local-restriction}: a specific property of the dependency graph of each voter (and not of their union). More precisely, we aim at the identification of a conveniently describable restriction that will be sufficient to ensure that \cav\ can be efficiently solved. One straightforward idea would be to explore whether requiring each voter to submit a graph with bounded treewidth would suffice towards applying the positive result from \Cref{characterization}. This is not the case. The hardness result from \cite{bl16} shows that even with voters' dependency graphs of treewidth $1$ (a single edge per voter), the problem is NP-hard. 
Moreover, even with just two voters, 
the positive result of \Cref{characterization} cannot be applied
under the restriction of bounded treewidth (or even pathwidth) per voter.

\begin{observation}
\label{thm:paths}
Bounding the pathwidth of voters' dependency graphs by a constant doesn’t imply constant treewidth for the global dependency graph, even for 2 voters.
\end{observation}
\begin{proof}
Suppose each voter casts a dependency graph that consists of $\rho$ paths, each of length $\rho,$ where $\rho$ equals $\sqrt{m}$. Specifically, the first voter casts paths along the sets of issues $\{1, 2, \ldots, \rho\}$, $\{\rho + 1, \rho + 2, \ldots, 2\rho\}$, and so on, up to $\{(\rho-1)\rho + 1, (\rho-1)\rho + 2, \ldots, \rho^2\}$. The second voter casts paths along the sets of issues $\{1, \rho + 1, 2\rho + 1, \ldots, (\rho-1)\rho + 1\}$, $\{2, \rho + 2, 2\rho + 2, \ldots, (\rho-1)\rho + 2\}$ and so on, up to $\{\rho, 2\rho, 3\rho, \ldots, \rho^2\}$. The pathwidth (and in turn the treewidth) of their dependency graphs equals $1$. The resulting global dependency graph forms a $ \sqrt{m} \times \sqrt{m}$ grid and the treewidth of such grids is non-constant \cite{robertson1986graph}. 
\end{proof}

Therefore, 
stricter restrictions are necessary, even for the case of two voters, if one wishes to solve \cav\ by utilizing the algorithm from \cite{mp21}. We identify the bounded \textit{vertex cover number} of the voters' dependency graphs as the necessary and sufficient condition to ensure bounded treewidth of the global dependency graph, thus enabling polynomial-time algorithms, at least for the case of a constant number of voters. The vertex cover number of a graph is a natural and well studied graph parameter in parameterized algorithmics (cf.~\cite{fellows2008graph,fiala2011parameterized,eiben2018structural,lampis2012algorithmic,fellows2009complexity}), but, notably, our work is one of the very few in the voting literature to explicitly consider this as a refinement of treewidth, aiming to provide positive results for cases where the treewidth restriction is not well suited. Additionally, it is a parameter that is one of the largest, right after the number of vertices, among several that are incomparable with each other. This makes it a reasonable restriction when we wish to avoid limiting the number of issues in \cav. Importantly, it is easy to define for a broad audience, making it feasible to enforce voting with such a restriction in mind; unlike treewidth. Finally, determining whether a voter's graph meets the criterion of having a bounded vertex cover number is straightforward, simplifying implementation and evaluation of the framework and the rule respectively.

\begin{theorem}
\label{thm:VC}  
If the vertex cover number of voters' dependency graphs is bounded by a constant and $\Delta=1,$ then \cav\ can be optimally solved for a constant number of voters. 
\end{theorem}
\begin{proof}
Let us first consider two voters, each with a dependency graph of bounded vertex cover number. We call $C$ the vertex cover of the dependency graph of the first voter. The dependency graph of the second voter has bounded treewidth (as a consequence of the fact that it has a bounded vertex cover number), which implies it has a tree decomposition of bounded width, say $\mathcal{T}$. We now focus on their global dependency graph and we create a tree decomposition for it as follows: we start with the decomposition $\mathcal{T}$ and we augment each of its bags by adding the vertex cover $C$. Since the size of each bag was constant to begin with, adding a constant number of vertices (from \( C \)) keeps the bag size constant. Therefore, the global dependency graph, in the case of two voters, has bounded treewidth.

To prove the statement for any constant number of voters we can continue this process by iteratively unioning the emerged graph with the dependency graph of each subsequent voter. This procedure corresponds to forming the union of, once again, a graph of bounded treewidth with a graph of bounded vertex cover.
Doing so for a constant number of iterations proves that the union of the dependency graphs of these voters results in a graph with bounded treewidth. The efficient algorithm from \Cref{characterization} can then be used.
\end{proof}

Notably, \Cref{thm:VC} works not only for a constant number of voters but also when the set of distinct dependency graphs cast by the voters are constantly many, regardless of the number of the voters. On the contrary, it cannot be extended to handle a non-constant number of voters, as this would contradict the hardness established in \cite{bl16}, where each voter had, actually, a dependency graph of bounded vertex cover number. 

A slight modification in the proof of \Cref{thm:paths} (splitting each voter $i$ into two, with one taking the even edges and the other taking the odd edges of each path in $i$'s dependency graph) shows that even if the cast graphs are matchings, the polynomial algorithm from \Cref{characterization} cannot be used. To exclude matchings, one would need to bound the vertex cover number of the graph, indicating tightness of \Cref{thm:VC} with respect to applying the only known in the prior literature algorithm for \cav.

We conclude by referring to concise families of instances that satisfy the restrictions stated in \Cref{thm:VC} and could naturally occur in elections with interdependent issues. One example concerns instances with a constant number of dependencies (i.e., edges in each voter's dependency graph), along with an unrestricted number of unconditional ballots per voter. The most intuitive family of instances for which the theorem applies involves those with a constant number of out-stars per voter. These dependency graphs could arise in practice when each voter specifies a few complex and critical issues, significantly impacting a broader set of other issues. Realistically, this process could be implemented through the following simple steps: First, a voter $i$ identifies a small set of projects or candidates $S_i$ that are important to her and she is knowledgeable about. Next, she specifies mutually disjoint sets of other projects or candidates that are influenced by her decisions regarding issues in $S_i$. Finally, for some alternatives of the issues in $S_i$, the voter indicates their preferred alternatives for the corresponding issues in the selected sets from the previous step (voting unconditionally for any other issue). The mutual disjointness naturally appears in real-life scenarios when projects are partitioned into groups, with those in one group not affecting those in other groups. For example, in a participatory budgeting scenario, this could correspond to projects in different neighborhoods, where voters do not relate issues in one district to those in another, or to projects of different types that voters do not associate with one another, or to projects planned for different time periods. Yet, notably, such a process, while sufficient to create the structure of out-stars, is not necessary for a bounded vertex cover number. 

\section{Conclusions and Future Work}

This work examines the computational complexity of computing the outcome of Conditional Minisum Approval Voting rule for multi-issue elections on interdependent issues, aiming to identify when and how it can be practically applied. The first part of the work highlights a fundamental challenge: while the setting allows voters to express nuanced preferences, this flexibility leads to severe computational intractability, specifically demonstrating that the relevant problem is prohibitively hard, even for restricted instances. In response, the second part introduces practical limitations on voter ballots and dependency relationships that achieve to ensure efficient computation of outcomes without significantly sacrificing expressivity, making \cav\ viable for real-world use. Overall, our work effectively and completely resolves the question of efficient computation of the examined rule, paving the way for adapting and analyzing further classical voting rules to combinatorial settings. This approach need not be confined to multi-issue contexts but could extend to committee elections or participatory budgeting. Beyond this, we see two main research directions in voting under combinatorial preferences: (a) examining proportionality guarantees in committee elections and participatory budgeting frameworks where voters’ preferences are interdependent over candidates, and (b) collecting suitable real-world data and performing experimental evaluations.

\vspace{1.2cm}

\noindent \textbf{Acknowledgments.}
G.~Amanatidis and E.~Markakis were partially supported by project MIS 5154714 of the National Recovery and Resilience Plan Greece 2.0 funded by the European Union under the NextGenerationEU Program. M.~Lampis was partially supported by ANR project ANR-21-CE48-0022 (S-EX-AP-PE-AL).
G.~Papasotiropoulos is supported by the European Union (ERC, PRO-DEMOCRATIC, 101076570). Views and opinions expressed are however those of the author(s) only and do not necessarily reflect those of the European Union or the European Research Council. Neither the European Union nor the granting authority can be held responsible for them.
\begin{figure}[h!]
\centering \includegraphics[width=0.45\linewidth]{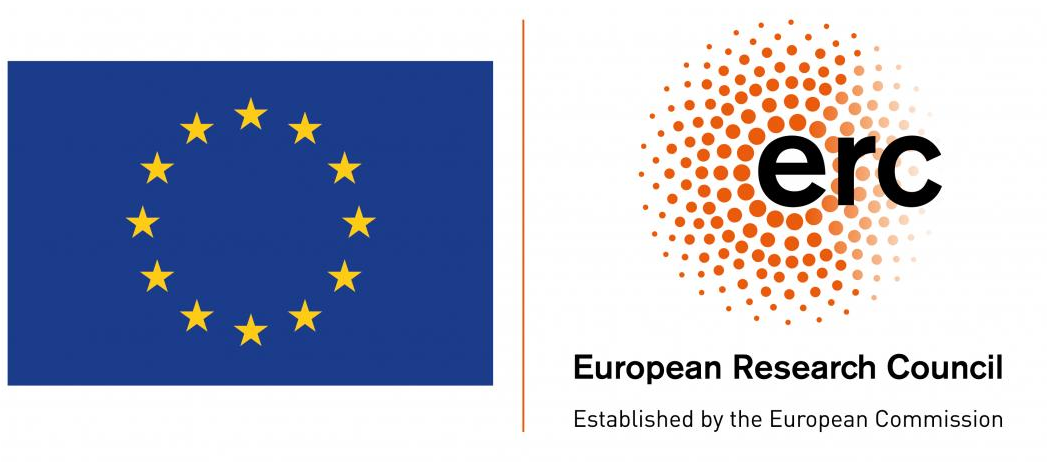}
\end{figure}

{\normalsize{
\bibliographystyle{alpha}
\bibliography{sample}
}}

\end{document}